\documentclass[letterpaper, 10 pt, conference]{ieeeconf}

\usepackage{amsthm}
\newtheorem{theorem}{Theorem}[section]

\newtheorem{proposition}[theorem]{Proposition}

\newtheorem{definition}[theorem]{Definition}
\newtheorem{assumption}[theorem]{Assumption}

\usepackage{dblfloatfix}
\usepackage{amssymb}
\usepackage{titlesec}
\titlespacing*{\subsection}{0pt}{0.5ex}{0.3ex}
\usepackage{mathtools}
\usepackage{paralist}
\usepackage{color}
\usepackage{bm}
\usepackage[normalem]{ulem} 
\usepackage{booktabs}
\usepackage{threeparttable}

\IEEEoverridecommandlockouts

\title{\LARGE \bf
Salted Fisher Information for Hybrid Systems
}

\author{Bukunmi G. Odunlami, Marcos Netto, and Hai Lin
\thanks{This work is supported by the National Science Foundation under Grant 2328241. B. G. Odunlami and M. Netto are with the Department of Electrical and Computer Engineering, New Jersey Institute of Technology, Newark, NJ 07102, USA. H. Lin is with the Electrical Engineering Department, University of Notre Dame, Notre Dame, IN 46556, USA.}
}

\begin{document}

\maketitle
\thispagestyle{empty}
\pagestyle{empty}
\overrideIEEEmargins

\begin{abstract}
Discrete events change how parameter-influence propagates in hybrid systems. Prevailing Fisher information formulations assume that sensitivities evolve smoothly according to continuous-time variational equations and therefore neglect the sensitivity updates induced by discrete events. This paper derives a Fisher information matrix formulation compatible with hybrid systems. To do so, we use the saltation matrix, which encodes the first-order transformation of sensitivities induced by discrete events. We call the resulting formulation the salted Fisher information matrix (SFIM). The proposed framework unifies continuous information accumulation during flows with discrete updates at event times. We also show that hybrid persistence of excitation is sufficient for the SFIM to be positive definite.
\end{abstract}

\section{Introduction}
Fisher information quantifies how parameter sensitivities determine bounds on estimator variance and therefore provides a basis for assessing parameter identifiability \cite{VirginFisherPowerSystems}. It has been extensively studied in signal processing and quantitative experiment design \cite{Sovljanski2024}, with growing applications in power system modeling \cite{9585016}.

The Fisher information is computed along trajectories of a dynamical system using parameter sensitivities, assuming that the state and its sensitivities evolve smoothly over time \cite{JauffFisher}. However, many engineered systems are hybrid and involve discrete events triggered by conditions such as protection mechanisms or control thresholds \cite{odunlami2025hybrid}. In such hybrid systems, parameter influence propagates through both continuous flows and discrete events, altering the evolution of system trajectory sensitivities with respect to parameters and making parameter identifiability nontrivial to assess \cite{observabilityofHybrid}.


An extension of the Fisher information matrix (FIM) to \emph{switching systems}, a particular class of hybrid systems, treats mode transitions as reparameterizations and aggregates intra-mode sensitivities \cite{Cavicchioli2017}. That is, sensitivities are computed within each mode based on the corresponding continuous dynamics, and then combined across modes without accounting for the effects of discrete events. However, in hybrid systems, parameter perturbations affect not only intra-mode dynamics but also event timing and reset geometry, reshaping the FIM through its quadratic dependence on sensitivities, cf. \eqref{eq.FIM}.

Fisher information has also been used to assess estimation performance in hybrid systems using discrete-time filtering formulations. Koutsoukos \cite{Koutsoukos2003} addresses discontinuities in state evolution by switching between continuous-time models and updating measurements from discrete sensors, thereby computing Fisher information from a sequence of sensor readings. His approach does not explain how parameter perturbations affect the system trajectory at the instant a discontinuity (discrete event) occurs. Recursive formulations of the FIM for nonlinear dynamical systems, such as the posterior Cram{\'e}r–Rao bound in \cite{Tichavsky1998}, assume smooth state evolution, with the trajectory's dependence on parameters continuous in time and governed solely by differential equations. These formulations do not apply when discrete events introduce jumps in the state and its sensitivities.


A natural question is how trajectories behave across discrete events. It might be assumed that nearby trajectories can be propagated across a discrete event by simply applying the Jacobian of the reset map \cite{Hartley}. However, when a trajectory crosses a guard, the post-event behavior depends not only on the reset map but also on the distinct dynamics in each mode \cite{KON2021109752}. Accordingly, we must account for both effects, as captured by the saltation matrix \cite{KONG2021109752}.

We use the saltation matrix to characterize how the state and its parameter dependence change when a system jumps---that is, switches from one mode to another, where each mode $q$ corresponds to a distinct set of continuous dynamics of the form $\dot{x}=f_q(x)$. The saltation matrix provides a first-order approximation of the jump effect based on the dynamics in each mode, the reset map that describes how the state and parameters update at the jump, and the guard conditions that determine when the jump occurs.

We derive the evolution of sensitivity along hybrid arcs, including changes at guard crossings. The resulting salted Fisher information matrix (SFIM) captures both continuous information accumulation during flows and discrete jump effects within a unified framework. To show how parameter changes affect switching behavior, we derive the sensitivity of event times to model parameters. We further show that hybrid persistence of excitation is a sufficient condition for the SFIM to be positive definite.

\noindent
\emph{Notation.} A superscript $\boldsymbol{\cdot}^-$ $(\boldsymbol{\cdot}^+)$ denotes a variable or quantity immediately before (after) a discrete event. Given a differentiable scalar- or vector-valued function $\phi$, $D_x\phi$, $D_\theta\phi$, and $D_t\phi$ denote differentiation of $\phi$ with respect to $x$, $\theta$, and $t$, respectively. $\land$ is the logical AND operator, $\|\boldsymbol{\cdot}\|$ denotes the Euclidean norm, and $\sigma_{\min}(\boldsymbol{\cdot})$ the smallest singular value of its matrix argument. For a symmetric matrix $A$, $\lambda_{\min}(A)$ is its smallest eigenvalue and $A\succ 0\, (A \succeq 0)$ denotes that $A$ is positive definite (semidefinite).

 

\section{Preliminaries}
\subsection{Maximum likelihood estimation}
Let $y=\{y_i\in \mathbb{R}\}_{i=1}^{m}$ be a set of independent and identically distributed (i.i.d.) observations. Assume that these observations follow a joint probability density function $p(y;\theta)$ for a given parameter $\theta\in\mathbb{R}$. We define the likelihood function as $\mathcal{L}(\theta;y)\coloneqq p(y;\theta)$. Since the observations are independent, the joint density equals the product of the marginal densities. That is,
\begin{equation*}
\mathcal{L}(\theta;y) \coloneqq p(y;\theta)
= \prod_{i=1}^{m}p(y_{i};\theta).
\end{equation*}

Note that the likelihood function is a function of $\theta$ for fixed $y$. We seek the numerical value of $\theta$ that best explains the observations, that is, the value of $\theta$ that maximizes $\mathcal{L}(\theta;y)$. An estimator that maximizes $\mathcal{L}(\theta;y)$ is termed a maximum likelihood estimator,
\begin{equation*}
\widehat{\theta}_{\mathrm{ML}}\coloneqq\arg\max_{\theta}\; \mathcal{L}(\theta;y)
= \arg\max_{\theta}\; p(y;\theta).
\end{equation*}

For the exponential family of probability distributions (e.g., the Gaussian distribution), it is more convenient to minimize the negative of the natural logarithm of the likelihood function. Formally,
\begin{equation*}
\widehat{\theta}_{\mathrm{ML}}=\arg\max_{\theta}\; p(y;\theta) = \arg\min_{\theta}\; \left[-\ln p(y;\theta)\right].
\end{equation*}

Since the observations are independent, it follows that
\begin{equation}\label{eq.ML_estimator}
\widehat{\theta}_{\mathrm{ML}}= \arg\min_{\theta}\sum_{i=1}^{m} \left[-\ln p(y_{i};\theta)\right].
\end{equation}

Now, define the score function,
\begin{equation}\label{eq.score_function}
\rho(\theta;y) \coloneqq D_\theta\ln\mathcal{L}(\theta;y)
= D_\theta\ln p(y;\theta).
\end{equation}

From \eqref{eq.ML_estimator} and \eqref{eq.score_function}, an estimate is obtained by solving for $\rho(\theta;y)=0$. Now, let $\widehat{\theta}$ be an unbiased estimator, that is, $\mathbb{E}[\,\widehat{\theta}\,]=\theta_{\mathrm{true}}$, and $\sigma_{\widehat{\theta}}^{2}$ its variance. The latter is a measure of the accuracy of $\widehat{\theta}$. It is natural to inquire about the minimum attainable variance any estimator can achieve. An estimator that attains this lower bound would be considered optimal. To answer this question, one needs a measure of the amount of information that an observable $y$ carries about $\theta$ upon which the probability of $y$ depends. The Fisher information provides this measure.


\subsection{Fisher information}
We define the Fisher information \cite{Kay1993} as the variance of the score function,
\begin{align*}
\mathcal{I}(\theta) &\coloneqq \mathbb{E}\left[\left( D_\theta\ln\mathcal{L}(\theta;y) \right)^{\!2}\right] \\ 
&= \int_{-\infty}^{\infty}\left( D_\theta\ln p(y;\theta) \right)^{\!2} \!p(y;\theta)\,dy,
\end{align*}

\noindent
and $\mathcal{I}(\theta)\ge 0$. When there are $n$ parameters, $\theta\in\mathbb{R}^n$, the Fisher information takes the form of a matrix---the FIM, which is given by
\begin{equation}\label{eq.Fisher.vector.theta}
\mathcal{I}(\theta) = \mathbb{E}\left[\left(D_\theta\ln p(y;\theta)\right)\left(D_\theta\ln p(y;\theta)\right)^{\!\top}\right].
\end{equation}

Now, assume we observe a signal $y(t)$ that consists of a deterministic model $h(t,\theta)$ and additive noise,
\begin{equation}\label{eq.meas.model}
y(t) = h(t;\theta)+\epsilon(t),
\end{equation}
where $\epsilon(t)\sim\mathcal{N}(0, V)$. Recall that the multivariate Gaussian distribution has a probability density function given by
\begin{align}\label{eq.multivar.Gaussian}
p(y;\theta) &= \left[\left(2\pi\right)^{n}\det\left(V\right)\right]^{-\frac{1}{2}} \nonumber \\
&\;\;\;\;\;\;\;\exp\left(-\frac{1}{2}\left(y-h(t;\theta)\right)^{\top}V^{-1}\left(y-h(t;\theta)\right) \right).
\end{align}

From \eqref{eq.multivar.Gaussian}, the log-likelihood function
\begin{equation}\label{eq.x1}
\ln p(y;\theta)\propto -\frac{1}{2}\left(y-h(t;\theta)\right)^{\top}V^{-1}\left(y-h(t;\theta)\right).
\end{equation}

Now, define $u\coloneqq y-h(t;\theta)$. Then,
\begin{align}\label{eq.x2}
D_\theta\left(u^{\top}V^{-1}u\right)
=2\left(D_\theta u\right)V^{-1}u.
\end{align}

It follows from \eqref{eq.meas.model}, \eqref{eq.x1}, and \eqref{eq.x2} that
\begin{align}
D_\theta \ln p(y;\theta) &= \frac{1}{2}\cdot2\cdot J(t)^{\top}V^{-1}\left(y-h(t;\theta)\right) \nonumber \\
&=J(t)^{\top}V^{-1}\epsilon(t) \label{eq.x3}
\end{align}

\noindent
where $J(t) = D_\theta h(t;\theta)$. Substituting \eqref{eq.x3} into \eqref{eq.Fisher.vector.theta} yields
\begin{align*}
\mathcal{I}(\theta) &= \mathbb{E}\left[\left(J(t)^{\top}V^{-1}\epsilon(t)\right)\left(J(t)^{\top}V^{-1}\epsilon(t)\right)^{\!\top}\right] \nonumber \\
&= \mathbb{E}\left[J(t)^{\top}V^{-1}\epsilon(t)\epsilon(t)^{\top}V^{-1}J(t)\right] \nonumber \\
&= J(t)^{\top}V^{-1}\mathbb{E}\left[\epsilon(t)\epsilon(t)^{\top}\right]V^{-1}J(t) \nonumber \\
&= J(t)^{\top}V^{-1}VV^{-1}J(t) \nonumber \\
&= J(t)^{\top}V^{-1}J(t),
\end{align*}

\noindent
where we used the property $\left(ABC\right)^{\top}=C^{\top}B^{\top}A^{\top}$, and the fact that $J(t)$ and $V$ are deterministic. When collecting data continuously over a time interval $\left[0, T\right]$, the total information is given by
\begin{equation}\label{eq.FIM}
\mathcal{I}(\theta) = \int_{0}^{T} J(t)^\top V^{-1} J(t)\,dt.
\end{equation}


Eq. \eqref{eq.FIM} is fundamental to the salted Fisher information we propose in \S\ref{sec:SFIM}. Note that for any unbiased estimator $\widehat{\theta}$ of the parameter vector $\theta$, the estimation covariance is bounded below by the reciprocal of the FIM, $\mathrm{Cov}(\widehat{\theta}) \succeq \mathcal{I}(\theta)^{-1}$. Thus, $\mathcal{I}(\theta)^{-1}$ provides a lower bound on the achievable estimation variance. In particular, parameter identifiability requires $\mathcal{I}(\theta) \succ 0$ \cite{6638944}.

\subsection{Hybrid systems}
Hybrid systems combine continuous dynamics and discrete events that interact and jointly govern the system's evolution over time. The continuous dynamics are referred to as \emph{flows}, and discrete events are called \emph{jumps} or \emph{transitions}. Discrete events are triggered by conditions such as guard crossings and cause the system to change from one mode to another. Following \cite{Antsaklis}, we define a hybrid system
\begin{equation*}
\mathcal H \coloneqq \langle Q, X, f, \mathrm{Init}, \mathrm{Inv}, G, R \rangle,
\end{equation*}
where $Q=\{q_1,\dots,q_N\}$ is a finite set of discrete modes and $X\subseteq\mathbb R^n$ is the continuous state space. For each $q\in Q$, $f_q:Q\times X \to\mathbb R^n$ governs the continuous dynamics, and $\mathrm{Inv}_q\subset X$ is the invariant set. The guard set $G_{q q'}\subset X$ triggers transitions from mode $q$ to mode $q'$, and $R_{q q'}$ is the reset map. A hybrid arc satisfies
\begin{equation*}
\begin{cases}
\dot x = f_q(x) & x \in \mathrm{Inv}_q, \\
x^+ = R_{q q'}(x^-) & x^- \in G_{q q'}.
\end{cases}
\end{equation*}

\subsection{Saltation matrix}
In hybrid systems, linearization does not extend directly through mode transitions. When a trajectory reaches a switching surface and undergoes a reset, nearby trajectories are no longer synchronized with the nominal one. In particular, state variations induce variations in the switching time through the guard condition.

A nearby trajectory may spend slightly more or less time evolving under the pre-transition vector field $f_q$, and correspondingly less or more time under the post-transition vector field \( f_{q'} \). As a result, the relative displacement between trajectories depends not only on the reset, but also on this difference in how long each vector field acts.

The saltation matrix provides an adequate first-order mapping across a transition by accounting for both effects---reflecting the change from \( f_q \) to \( f_{q'} \) along with the dependence of the switching instant on the state. The saltation matrix has been used to analyze periodic hybrid motions, stability properties, and trajectory optimization \cite{7039861, IVANOV1998677}.

Consider perturbations $\delta x$ of a hybrid arc. Along flows in a mode $q$, they satisfy the variational equation $\delta\dot{x}=D_x f_q(x)\,\delta x$. At a transition, however, the perturbation must also account for the shift \(\delta t\) in the event time. Linearizing the guard condition \(g_{qq'}(x)=0\) gives
\[
0 = D_x g_{qq'}^-\,\delta x^- + \bigl(D_t g_{qq'}^- + D_x g_{qq'}^- f_q^-\bigr)\delta t,
\]
so that
\begin{align}
\delta t
=
-\frac{D_x g_{qq'}^-\,\delta x^-}
{D_t g_{qq'}^- + D_x g_{qq'}^- f_q^-}.
\label{eq.7}
\end{align}
To determine $\delta x^+$, we evaluate the reset relation at the shifted even time
$
x^+(t+\delta t)=R_{qq'}\!\left(x^-(t+\delta t),\,t+\delta t\right).
$
Using
\[
x^+(t+\delta t)\approx x^+ + \delta x^+ + f_{q'}^+\,\delta t,
\]
and linearizing the reset map yields
\[
x^+ + \delta x^+ + f_{q'}^+\,\delta t
=
R_{qq'}^- + D_x R_{qq'}^- \bigl(\delta x^- + f_q^-\,\delta t\bigr)
+ D_t R_{qq'}^-\,\delta t.
\]
Since \(x^+=R_{qq'}^-\), it follows that
\begin{align}
\delta x^+
=
D_x R_{qq'}^-\,\delta x^-
+
\bigl(D_x R_{qq'}^- f_q^- + D_t R_{qq'}^- - f_{q'}^+\bigr)\delta t.
\label{eq.8}
\end{align}
Substituting \eqref{eq.7} for \(\delta t\) in \eqref{eq.8} gives
\[
\delta x^+
=
D_x R_{qq'}^-
+
\frac{
\bigl(f_{q'}^{+} - D_x R_{qq'}^- f_q^- - D_t R_{qq'}^-\bigr)
D_x g_{qq'}^-
}{
D_t g_{qq'}^- + D_x g_{qq'}^- f_q^-
}
\delta x^-.
\]
Following \cite{KONG2021109752}, the saltation matrix is derived as:
\begin{align}
\Xi_{q q'}
= D_x R_{q q'}^{-}
+
\frac{
\big(f_{q'}^{+} - D_x R_{q q'}^{-} f_q^{-} - D_t R_{q q'}^{-}\big)
D_x g_{q q'}^{-}
}{
D_t g_{q q'}^{-} + D_x g_{q q'}^{-} f_q^{-}
}
\label{eq:saltation}
\end{align}
\noindent
where
\begin{equation}\label{eq:regular_crossing}
D_t g^-_{q q'} + D_x g_{q q'} f^-_q \neq 0.
\end{equation}
Thus, the second term in \eqref{eq:saltation} is the correction induced by the event-time shift.

\subsection{Motivational example}
Consider a DC--DC buck converter \cite{5675187} (Fig.~\ref{fig:buck}). Depending on the switching behavior and load conditions, the converter may operate in either continuous conduction mode (CCM) or discontinuous conduction mode (DCM). In CCM, the inductor current remains strictly positive during each switching cycle. Conversely, in DCM, the inductor current may decrease to zero before the next switching cycle begins, creating an additional interval in which the inductor does not conduct. We consider the DCM.

In the DCM operation considered here, the three modes correspond to distinct circuit configurations. In $q_1$, the switch is closed, and the diode is reverse-biased, so the input source energizes the inductor. In $q_2$, the switch is open, and the diode conducts, so the inductor supplies the capacitor and load. In $q_3$, both the switch and diode are nonconducting, and the output capacitor supplies the load until the next switching cycle.

Let $q = q_1$ (switch closed) in Fig.~\ref{fig:buck}. Applying Kirchhoff's voltage law across the inductor loop yields $L \dot{i}_L = v_{in} - v_C - r_L i_L$; Kirchhoff's current law at the capacitor node yields $C \dot{v}_C = i_L - \frac{v_C}{R}$. One can similarly obtain expressions for $q = q_2$ and $q = q_3$. The state variables are $x_1=i_L$ (inductor current) and $x_2=v_C$ (capacitor voltage). The parameters are $\theta_1=L$ (inductance), $\theta_2=C$ (capacitance), and $\theta_3=r_L$ (parasitic winding resistance of the inductor). Let the input source voltage be a constant, $a = v_{in}$. We describe the DCM as follows:
\begin{equation*}
\begin{cases}
\dot x_1 = \dfrac{a - x_2 - \theta_3 x_1}{\theta_1}, \quad
\dot x_2 = \dfrac{x_1 - x_2/R}{\theta_2}
& \text{if } q = q_1,\\[6pt]
\dot x_1 = \dfrac{-x_2 - \theta_3 x_1}{\theta_1}, \quad
\dot x_2 = \dfrac{x_1 - x_2/R}{\theta_2}
& \text{if } q = q_2,\\[6pt]
\dot x_1 = 0, \quad
\dot x_2 = \dfrac{-x_2/R}{\theta_2}
& \text{if } q = q_3,
\end{cases}
\end{equation*}
where the guard sets with constant thresholds $\bar v>\underline v$ are
\begin{align*}
G_{q_1 q_2} &= \{x\in X : x_2 \ge \bar v \},\\
G_{q_2 q_1} &= \{ x \in X : x_2 \le \underline v \},\\
G_{q_2 q_3} &= \{ x \in X : x_1 = 0 \}.
\end{align*}

\begin{figure}[!t]
\centering
\includegraphics[width=0.8\linewidth]{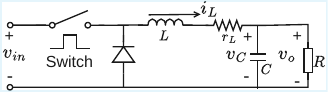}
\caption{One-line diagram of a buck converter.}
\label{fig:buck}
\end{figure}

\begin{figure}[!t]
\centering
\includegraphics[width=1\linewidth]{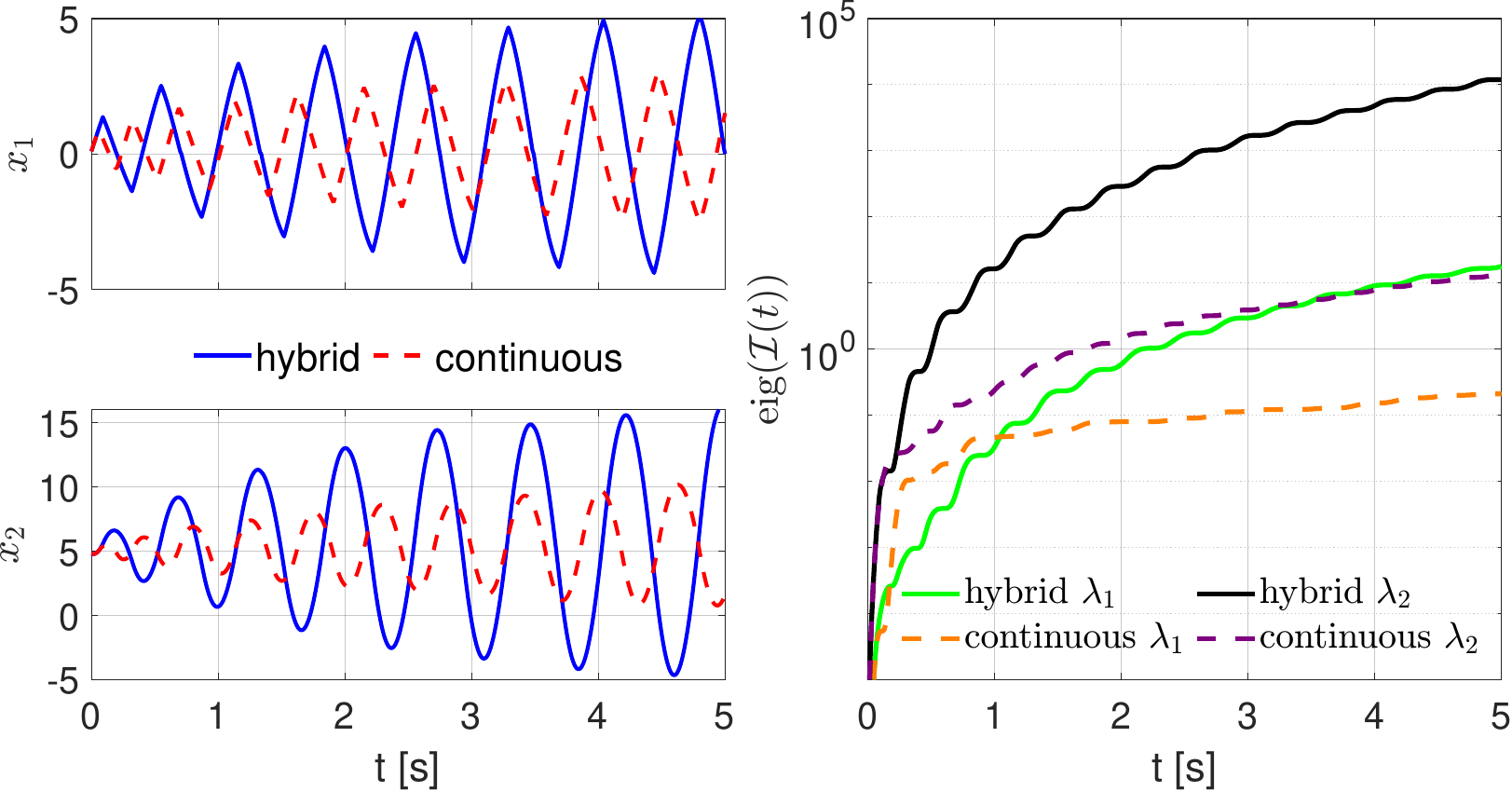}
\caption{Left: State trajectories obtained using the hybrid and continuous (averaged) models. Right: Corresponding FIM eigenvalues.}
\label{fig:buckconverter}
\end{figure}

The reset map $R_{q_2 q_3}=\{(x^-,x^+) \in X \times X: x_1^+=0 \land \ x_2^+=x_2\theta_1\}$ and identity otherwise. Given the measurement $y(t) = x_2(t) + \epsilon$ with $\epsilon \sim \mathcal{N}(0,\sigma^2)$, we compute the FIM using the proposed framework in Section~\ref{sec:SFIM} and compare it with the standard continuous averaged model. In the averaged model, the dynamics over a switching period are approximated by a weighted combination of the vector fields associated with the modes $q_1$, $q_2$, and $q_3$, 
\begin{equation*}
\dot{x} = d_1 f_{q_1}(x,\theta) + d_2 f_{q_2}(x,\theta) + d_3 f_{q_3}(x,\theta),    
\end{equation*}
where $d_1$, $d_2$, and $d_3$ denote fractions of the switching period spent in each mode, $d_1+d_2+d_3=1$. Fig. \ref{fig:buckconverter} compares the hybrid and averaged models along the same trajectory. For parameters $\theta_1$ and $\theta_2$, we examine the eigenvalues of the FIM, denoted by $\lambda_1(t)$ and $\lambda_2(t)$. These eigenvalues measure the information available for parameter estimation. The smallest eigenvalue corresponds to the least informative direction in the parameter space. For the averaged model, the eigenvalues increase gradually, and the smallest eigenvalue remains comparatively small over the horizon, which indicates limited information in at least one parameter direction. The averaged model omits the sensitivity updates associated with state resets at each jump. In contrast, propagating sensitivities through flows and updating them using the saltation matrix preserves switching effects and yields larger FIM eigenvalues.

\begin{figure}[!t]
\centering
\includegraphics[width=\linewidth]{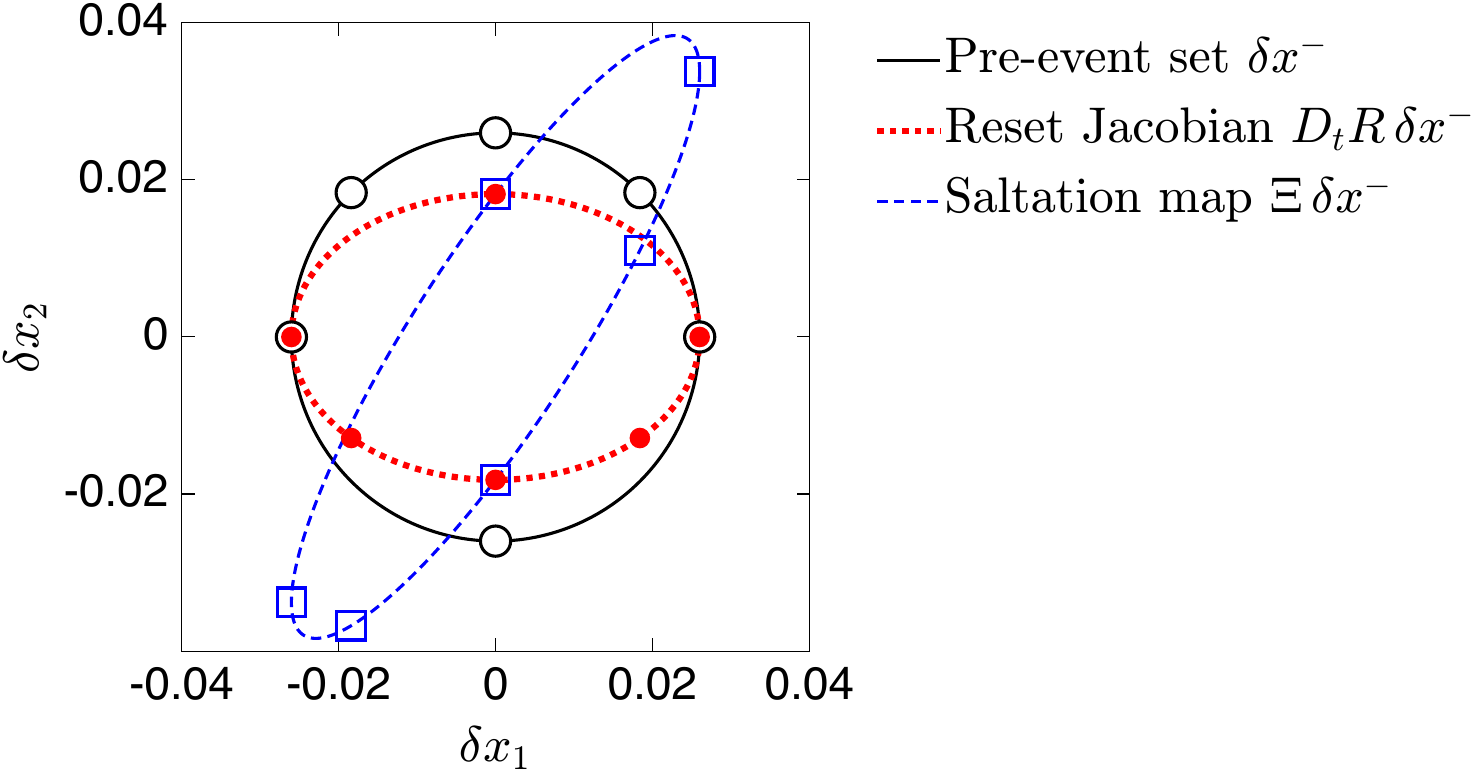}
\caption{The pre-event perturbation set, $\delta x^-$, mapped across a discrete event via the reset Jacobian, $D_tR$, and the saltation matrix, $\Xi$. Representative perturbations (unfilled black circles) and their images under $D_tR$ (red dots) and $\Xi$ (blue squares) highlight the induced geometric deformation.}
\label{fig:saltationVSrest}
\end{figure}

To further illustrate the role of the saltation matrix, consider a trajectory crossing a switching surface defined by \( g(x) = x_1 \), with \( D_x g = [1\;0]^{\top} \). Let \( \delta x^- \) denote a set of pre-event deviations around the nominal crossing point. If the transition is treated using only the reset map, the corresponding first-order update is
$
\delta x^+ = D_x R\,\delta x^-,
$
In contrast, the saltation-based update is given by
$
\delta x^+ = \Xi\,\delta x^-.
$
Let the pre- and post-event vector fields at the switching point be given by \( f^- = [1\; -1]^{\top} \), \( f^+ = [1\; 2]^{\top} \), and a reset Jacobian
\[
D_x R = \begin{bmatrix} 1 & 0 \\ 0 & -0.7 \end{bmatrix}.
\]

As shown in Fig~\ref{fig:saltationVSrest}, a circular pre-event set of radius \( 0.026 \) is mapped across the event using both updates. While \( D_x R \) produces a simple linear transformation of the perturbation set, the saltation matrix yields a distinct deformation reflecting the combined effect of the reset and the change in dynamics.

\section{Salted Fisher Information} \label{sec:SFIM}
We now describe how the states' sensitivities change along a hybrid arc with respect to parameters.

\begin{assumption}[Regularity]\label{ass.3.1}
We consider piecewise-smooth hybrid systems with differentiable vector fields, reset maps with transversal guard crossings, and finitely many transitions on compact time intervals \cite{HPE}.
\end{assumption}
\begin{proposition}[Hybrid sensitivity evolution]\label{prop:hybrid_sensitivity} The matrix containing the sensitivities of the state variables with respect to parameters, $Z(t):=D_\theta x(t,\theta)$, evolves as follows. For a mode $q \in Q$, while the dynamics are flowing,
\begin{equation*}
\dot Z(t) = D_x f_q(x(t),\theta)\, Z(t) + D_\theta f_q(x(t),\theta) \quad t \in (\tau_j,\tau_{j+1}).
\end{equation*}

At event time $\tau_j$, the state undergoes a reset
\begin{equation} \label{eq:resett}
x(\tau_j^+) = R_{qq'}(x(\tau_j^-),\theta).
\end{equation}

\noindent
Differentiating \eqref{eq:resett} with respect to $\theta$  while accounting for the parameter dependence of the event time through the saltation matrix yields the update
\begin{equation}\label{eq:final_Z}
Z(\tau_j^+) = \Xi_j\, Z(\tau_j^-) + D_\theta R_{qq'}(x(\tau_j^-),\theta),
\end{equation}

\noindent
where $\Xi_j$ is the saltation matrix associated with the $j$th transition.
\end{proposition}

The salted Fisher information is defined under Gaussian assumptions, with measurements available during flows and immediately after guard crossings. 

\begin{definition}[Salted Fisher information]\label{def:SFIM}
The Salted Fisher information is defined as
\begin{equation*}
\mathcal{I}_{\mathrm{S}}(\theta) \coloneqq \int_{0}^{T} J(t)^\top V^{-1} J(t)\, dt + \sum_{j} J(\tau_j^+)^\top V_j^{-1} J(\tau_j^+)
\end{equation*}
for a finite horizon $[0,\,T]$, where $\tau_j$ denotes guard-crossing event times, and
\begin{equation*}
J(t) = D_x h(x(t),\theta)\, Z(t) + D_\theta h(x(t),\theta)
\end{equation*}
is the output sensitivity; $Z(t)$ evolves according to the hybrid sensitivity dynamics of Proposition \ref{prop:hybrid_sensitivity}. The jump contribution is evaluated using the post-event sensitivity \eqref{eq:final_Z}.
\end{definition}

\subsection{Fisher information change at hybrid events}
We isolate the incremental change in information induced by a discrete event. At an event time $\tau_j$, the output sensitivity $J(t)$ may exhibit a discontinuity due to the jump in $Z(t)$ across the event. We quantify the resulting change in the quadratic information at the event as
\begin{equation}\label{eq:sfic}
\Delta \mathcal{I}_j = \big(J(\tau_j^+)\big)^\top V_j^{-1} J(\tau_j^+) - \big(J(\tau_j^-)\big)^\top V_j^{-1} J(\tau_j^-)
\end{equation}
which measures the instantaneous deformation of information geometry at the guard crossing. To isolate the saltation contribution, recall that the sensitivity update can be decomposed as
\begin{align*}
Z(\tau_j^+) = D_x R_{qq'} Z(\tau_j^-) + \big(\Xi_j - D_x R_{qq'}\big) Z(\tau_j^-)
\end{align*}
where the second term characterizes the induced deviation from the linearized reset map. The corresponding increment in output sensitivity is therefore
\begin{align*}
\Delta J_j \coloneqq D_x h(x(\tau_j)) \big(\Xi_j - D_x R_{qq'}(x(\tau_j)) \big) Z(\tau_j^-).
\end{align*}

Using $J(\tau_j^+) = J(\tau_j^-) + \Delta J_j$ in \eqref{eq:sfic} and expanding the quadratic form yields
\begin{align}
\Delta \mathcal{I}_j = 2 \big(J(\tau_j^-)\big)^\top V_j^{-1} \Delta J_j + (\Delta J_j)^\top V_j^{-1} \Delta J_j,
\end{align}

\noindent
which consists of a cross-term and a quadratic saltation contribution. The cumulative event contribution along a hybrid arc is
\begin{align}
\Delta \mathcal{I}_{\mathrm{tot}} = \sum_j \Delta \mathcal{I}_j.
\label{eq:totalI}
\end{align}

\noindent
The decomposition in \eqref{eq:totalI} describes how Fisher information accumulates along system trajectories and across discrete events, which is important for analysis and design of trajectories that maximize Fisher information for parameter estimation \cite{wilson2014fisher}.

\subsection{Parameter sensitivity of event times}
A hybrid arc may also consist of smooth solution segments separated at $\theta$-dependent switching times. Consequently, differentiating the trajectory with respect to $\theta$ requires differentiation of both the segment evolution and the breakpoints. Although saltation describes the first-order state deformation across a transition, the dependence of the switching instant itself on $\theta$ is not explicit in that representation. The following proposition gives the sensitivity of the switching time to $\theta$.

\begin{proposition}[Event-time sensitivity]\label{prop:event_time_sensitivity}
Suppose the event time $\tau(\theta)$ is defined implicitly by a guard condition
\begin{align}
g(x,\theta,\tau(\theta)) = 0.
\label{eq:guardId}
\end{align}
where $g:\mathbb{R}^n \times \mathbb{R}^p \times \mathbb{R} \to \mathbb{R}$ is continuously differentiable. Then, under the condition \eqref{eq:regular_crossing}, the derivative of the switching time with respect to $\theta$ is
\begin{align}
D_\theta \tau = - \frac{D_x g(\tau^-)\, Z(\tau^-) + D_\theta g(\tau^-)} {D_t g(\tau^-) + D_x g(\tau^-)\, f_q^-(\tau^-)}.
\label{eq:parametersensitivity}
\end{align}
\end{proposition}

\begin{proof}
Differentiating \eqref{eq:guardId} with respect to $\theta$ and applying the chain rule gives
\begin{align*}
D_x g(\tau^-) \frac{d}{d\theta}x(\tau(\theta),\theta)
+ D_\theta g(\tau^-)
+ D_t g(\tau^-)D_\theta \tau
=0.
\end{align*}
Using
$
\frac{d}{d\theta}x(\tau(\theta),\theta)
=
Z(\tau^-)+f_q^-(\tau^-)D_\theta \tau,
$
and solving for $D_\theta \tau$ yields \eqref{eq:parametersensitivity}.
\end{proof}

Proposition \ref{prop:event_time_sensitivity} shows that parameter perturbations affect both the state trajectory and the timing of the guard crossings. This timing variation alters the subsequent evolution of sensitivity and, in turn, the information accumulated for parameter identification. This motivates a structural excitation condition tailored to hybrid arcs.

\subsection{Hybrid persistence of excitation and SFIM full rank}
In parameter estimation, \emph{persistence of excitation} ensures that regressors carry sufficient independent information to uniquely identify all parameters \cite{Johnson}. Following \cite{HPE}, we define \emph{hybrid persistence of excitation (HPE)} as a condition combining both flow and event-level excitations. Formally, HPE requires that the saltation-aware regressors $\{J(t), J_j\}$ remain uniformly informative across hybrid windows.

\begin{definition}[Hybrid persistence of excitation]
Let the hybrid regressor Gramian on a window $[t_0,t_0+\mu_t]$ with jumps $j_0{:}j_0{+}\mu_j$ be
\begin{equation*}
G_{t_0,j_0} = \int_{t_0}^{t_0+\mu_t} J(t)^\top J(t)\,dt + \sum_{k=j_0}^{j_0+\mu_j} J_k^\top J_k.
\end{equation*}
The system is said to satisfy HPE if there exists $\alpha>0$ such that for all hybrid windows, $G_{t_0,j_0} \succeq \alpha I_p$ where $I_p \in \mathbb{R}^{p \times p}$ denotes the $p \times p$ identity matrix.
\end{definition}

The following result establishes that HPE is sufficient for the SFIM to be positive definite.
\begin{theorem}[HPE implies SFIM positive definite]
Suppose that the hybrid persistence of excitation condition holds on a hybrid horizon $\mathcal H$ covering all measurements used in
\[\mathcal I_{\mathrm{S}} = \int_{\text{flows in }\mathcal H} J(t)^\top V^{-1} J(t)\,dt + \sum_{\tau_j \in \mathcal H} J_j^\top V_j^{-1} J_j,\]
where $V \succ 0$ and $V_j \succ 0$. Then $\mathcal I_{\mathrm{S}} \succ 0$.
\end{theorem}

\begin{proof}
Let
\[\underline{\lambda} = \min\Big\{\lambda_{\min}(V^{-1}), \min_{\tau_j\in\mathcal H} \lambda_{\min}(V_j^{-1}) \Big\} > 0,\]

\noindent 
where $\lambda_{\min}(\cdot)$ is the smallest eigenvalue of its matrix argument. For any nonzero $v\in\mathbb{R}^p$,
\[v^\top \mathcal I_{\mathrm{S}} v = \int (J v)^\top V^{-1} (J v)\,dt + \sum (J_j v)^\top V_j^{-1} (J_j v).\]

\noindent
Since $V^{-1}$ and $V_j^{-1}$ are positive definite,
\[(J v)^\top V^{-1} (J v) \ge \lambda_{\min}(V^{-1}) \|J v\|^2\]
and similarly for each jump term. Hence

\begin{equation}\label{eq.21}
v^\top \mathcal I_{\mathrm{S}} v \ge \underline{\lambda} \left(\int \|J v\|^2 dt + \sum \|J_j v\|^2 \right).    
\end{equation}

The term in parentheses equals the quadratic form $v^\top G_{t_0,j_0} v$ of the hybrid regressor Gramian $G_{t_0,j_0}$. Since 
\[
\begin{aligned}
v^\top G_{t_0,j_0} v
&= v^\top \left( \int_{t_0}^{t_0+\mu_t} J^\top J\,dt
+ \sum_{j=j_0}^{j_0+\mu_j} J_j^\top J_j \right) v \\
&= \int_{t_0}^{t_0+\mu_t} (Jv)^\top (Jv)\,dt
+ \sum_{j=j_0}^{j_0+\mu_j} (J_j v)^\top (J_j v) \\
&= \int_{t_0}^{t_0+\mu_t} \|Jv\|^2\,dt
+ \sum_{j=j_0}^{j_0+\mu_j} \|J_j v\|^2,
\end{aligned}
\]

\noindent
from \eqref{eq.21}, 
$
v^\top \mathcal I_{\mathrm{S}} v \ge \underline{\lambda}\, v^\top G_{t_0,j_0} v.
$
By the HPE condition,
$v^\top G_{t_0,j_0} v \ge \alpha \|v\|^2. $
Combining the bounds gives
$v^\top \mathcal I_{\mathrm{S}} v \ge \underline{\lambda}\alpha \|v\|^2 $,
since $\underline{\lambda}\alpha>0$, the inequality holds for all $v\neq 0$, implying $\mathcal I_{\mathrm{S}} \succ 0$.
\end{proof}

\section{Numerical Experiment}
We apply the proposed SFIM framework to a three-bus power system shown in Fig.~\ref{fig:3bus_WTG}, adapted from \cite{Yu_reach}. The system consists of a synchronous generator connected to Bus 1, a wind turbine generator (WTG) at Bus 2, and a load at Bus 3, interconnected via lossless transmission lines.

\subsection{Differential-algebraic equations}
We model the synchronous generator using the Flux-Decay Model \cite{sauer_pai}, supplemented by basic models of the turbine-governor and excitation systems. That is,
\begin{equation}
\begin{aligned}
T^{\prime}_{d0}\dot{E}^{\prime}_{q} &= - E^{\prime}_{q} - (X_{d} - X^{\prime}_{d}) I_{d} + E_{\mathrm{fd}}\\
\dot{\delta} &= \omega - \omega_{s}\\
M\dot{\omega} &= T_{m} - E^{\prime}_{q} I_{q} - (X_{q} - X^{\prime}_{d}) I_{d} I_{q} - D(\omega - \omega_{s})\\
T_{\mathrm{SV}}\dot{T}_{m} &= - T_{m} + P_{\mathrm{c}} - (1/R_{d})\left((\omega/\omega_{s}) - 1\right)\\
T_{E}\dot{E}_{\mathrm{fd}} &= - E_{\mathrm{fd}} + k_{A} (V_{\mathrm{ref}} - V_{t})\\
X^{\prime}_{d}I_{d} &= E^{\prime}_{q} - V_1\cos(\theta_1-\delta)\\
X_{q}I_{q} &= -V_1\sin(\theta_1-\delta)
\label{eq:differentialStates}
\end{aligned}
\end{equation}
\begin{figure}[!t]
\centering
\includegraphics[width=1\linewidth]{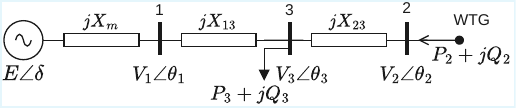}
\caption{Three-bus system with a wind turbine generator.}
\label{fig:3bus_WTG}
\vspace{-4mm}
\end{figure} 

\noindent
where $P_{\mathrm{c}}$ denotes the power command and $V_{\mathrm{ref}}$ the voltage reference; $T^{\prime}_{d0}$, $X_{d}$, $X^{\prime}_{d}$, $X_{q}$, $M$, $D$, $T_{\mathrm{SV}}$, $R_{d}$, $T_{E}$, and $k_{A}$ are parameters.

The network algebraic variables are determined by the AC power flow equations. For instance, the injected active power from the WTG at Bus~2 satisfies
\begin{equation*}
P_{\mathrm{ele}} - Y_{23}V_2V_3\sin(\theta_2-\theta_3)=0,
\end{equation*}
which couples the wind turbine power output to the network voltage magnitudes and phase angles. The quantities $V_2$, $V_3$, $\theta_2$, and $\theta_3$ denote the voltage magnitudes and phase angles at Buses 2 and 3, respectively; $Y_{23}$ is the admittance between the two buses, and $P_{\mathrm{ele}}$ corresponds to the active power injected by the WTG at Bus 2 under unity power factor.

\noindent
The WTG is modeled by a reduced-order aerodynamic and mechanical model. Let $z$ denote the turbine rotor speed and $\gamma$ the wind speed. The rotor speed dynamics of the form $\dot{z}=f_z(z,\beta,\gamma)$ and power output are given by
\begin{equation}
\begin{aligned}
\dot{z} &= \frac{\omega_{s}}{M_r} \left(\frac{B\, C_p(\zeta,\beta)\, \gamma^3}{z} - C z^2\right) \\
v &= P_2 = \kappa\, C_p(\zeta,\beta)\, \gamma^3
\label{eq:WTG}
\end{aligned}
\end{equation}
where $C_p(\zeta,\beta)$ denotes the aerodynamic power coefficient and $\beta$ is the pitch angle with parameters $M_r$, $B$, $C$, $\zeta$ and $\kappa$.


\subsection{Hybrid dynamics}
Hybrid behavior in \eqref{eq:differentialStates}--\eqref{eq:WTG} arises from the pitch activation mechanism, which adjusts the blade pitch angle to regulate aerodynamic torque and maintain the injected electrical power within its rated limit. This results in switching between pitch-inactive and pitch-active dynamics. Accordingly, the switching behavior is characterized by two guard conditions defined as follows:

\subsubsection{Rotor speed guard} When $z$ exceeds its rated value $z^{*}$, aerodynamic torque must be reduced to prevent mechanical overload. Pitch control is therefore activated, modifying the turbine vector field. Following \cite{WANG20111288}, the pitch controller dynamics are given by
\begin{equation}\label{eq.25}
\begin{aligned}
\dot{\xi} &= z - z^{*},\quad \dot{\beta} = \frac{1}{T_\beta} \Big(k_p (z - z^{*}) + k_i \xi - \beta\Big),
\end{aligned}
\end{equation}
where $\xi$ is the integral state, $k_p$ and $k_i$ are the proportional-integral gains, and $T_\beta$ is the pitch actuator time constant. From \eqref{eq:differentialStates}--\eqref{eq.25}, the hybrid system is given by
\begin{equation}
\mathcal H =
\left\{
\begin{aligned}
&Q = \{q_1,q_2\}\\
&X = \{E^{\prime}_{q},\, \delta,\, \omega,\, T_{m},\, E_{\mathrm{fd}},\, z,\, \xi,\, \beta\} \in \mathbb{R}^8 \\[1mm]
&f(q,X,u) =
\begin{cases}
\dot{x},\dot{z},0, 0 &\qquad q=q_1\\
\dot{x},\dot{z},\dot{\xi},\dot{\beta} &\qquad q=q_2
\end{cases} \\
&G_{q_1 q_2}=\{ X : z - z^{*} = 0,\ \dot z > 0 \}\\
&G_{q_2 q_1}=\{ X : z - z^{*} = 0,\ \dot z < 0 \}\\
&R_{q_1 q_2}(X)=R_{q_2 q_1}(X) = X.
\end{aligned}
\right.
\label{eq:hybridModel}
\end{equation}

\subsubsection{Power guard}
Now consider a power-triggered switching mechanism; pitch regulation is tied to the injected electrical power rising to its rated value. With a threshold $P_{\mathrm{ele}}^{*}$, the guard function is defined as
\begin{equation}
g_{P_{\mathrm{ele}}}(x,\theta,t) = P_{\mathrm{ele}}(x,\theta,t) - P_{\mathrm{ele}}^{*}.
\label{eq:powerGuard}
\end{equation}
Aside from the different guard condition, the hybrid dynamics in \eqref{eq:hybridModel} hold.

\subsection{Numerical simulations}
The proposed SFIM framework is implemented on the three-bus WTG power system described earlier to identify a subset of parameters $\theta = \{ M_r,\; \kappa,\; k_p,\; k_i,\; T_\beta,\; M,\; T_{E} \}$, with measurements $y = \{ E_{\mathrm{fd}},\, V_1,\, \delta,\, P_{\mathrm{ele}}\}$. The numerical values of the three-bus network, synchronous-generator, WTG, and pitch-controller parameters used in the simulations are summarized in Table~\ref{tab:system_parameters}. The system is simulated from nominal conditions with wind input $\gamma(t)=r+s\sin(100\pi t)$ inducing guard crossings, where $r$ is the mean wind speed and $s$ the oscillation amplitude. The simulations are performed over a $40$~s horizon using a fourth-order Runge--Kutta scheme with a time step of 0.005 s. Fig. \ref{fig:fim_switch} shows the resulting hybrid mode sequence as the system switches between $q_1$ and $q_2$, and compares the evolution of $\log \det(F(t)+\varepsilon I)$ for the smooth FIM and the SFIM, with SFIM exhibiting significantly greater information accumulation. For clarity, the comparison baselines are defined as follows. The smooth FIM propagates sensitivities without event-induced sensitivity updates. The reset FIM propagates sensitivities across transitions using only the reset Jacobian $D_xR$, whereas the SFIM uses the full saltation matrix and therefore accounts for the effects of the reset, changes in the vector field, and event timing on the sensitivity evolution.

\begin{table}[!t]
\centering
\caption{Parameter values for the 3-bus system and WTG model}
\label{tab:system_parameters}
\setlength{\tabcolsep}{4pt}
\begin{tabular}{c c | c c | c c | c c | c c}
\toprule
$\beta_{\mathrm{opt}}$ & $0$    & $k_i$    & $0.9$             & $Q_2$     & $0$    & $T^{\prime}_{d0}$  & $5.89$   & $X_{23}$ & $0.15$ \\
$C_{p,\mathrm{opt}}$   & $0.45$ & $k_p$    & $2.5$             & $Q_3$     & $0.5$  & $V_{\mathrm{ref}}$ & $0.9$    & $X^{\prime}_d$ & $0.18$ \\
$D$                    & $0.04$ & $M$      & $\frac{1}{15\pi}$ & $R_D$     & $0.05$ & $\omega_s$         & $120\pi$ & & \\
$\kappa$               & $0.25$ & $M_r$    & $2.0$             & $T_\beta$ & $0.1$  & $X_d$              & $1.31$   & & \\
$k_{A}$         & $20$   & $P_c$    & $0.6$             & $T_E$     & $0.2$  & $X_q$              & $1.26$   & & \\
$k_\beta$              & $0.35$ & $P_3$    & $1.0$             & $T_{SV}$  & $0.2$  & $X_{13}$           & $0.1$    & & \\ 
\bottomrule
\end{tabular}
\end{table}

\begin{table}[!t]
\centering
\caption{Information metrics of the FIM}
\label{tab:fim_metrics}
\begin{threeparttable}
\setlength{\tabcolsep}{4pt}
\begin{tabular}{lcccc}
\toprule
Method & rank & $\lambda_{\min}$ & $\sigma$ & $\log\det(F+\varepsilon I)$   \tnote{$\ddag$}\\
\hline
Smooth FIM & 6 & $1.07\times10^{-1}$ & $0$ & $8.14$ \\
SFIM & 7 & $1.56\times10^{-1}$ & $3.95\times10^{-1}$ & $62.56$ \\
\bottomrule
\end{tabular}
\begin{tablenotes}
\item[$\ddag$] We add $\varepsilon \!=\! 10^{-14}$ for numerical regularization in the log-determinant.
\end{tablenotes}
\end{threeparttable}
\end{table}

Let $F$ denote the FIM evaluated at the final time. Table~\ref{tab:fim_metrics} shows that the smooth FIM is rank deficient, while the SFIM achieves full rank and significantly larger information metrics. Moreover, the smallest eigenvalue $\lambda_{\min}$, the information scale $\sigma$, and the log-determinant $\log\det(F+\varepsilon I)$ increase substantially, indicating improved parameter identifiability when hybrid sensitivity effects are incorporated.

To complement the intrinsic information metrics in Table~\ref{tab:fim_metrics}, we evaluate the predictive accuracy of the SFIM through a local Monte Carlo parameter-estimation study. We consider the parameter vector $\theta$ and perform 300 independent realizations under the Gaussian measurement model. For each realization, a local weighted least-squares estimator, equivalent to the maximum-likelihood estimator under Gaussian noise, is used to obtain a parameter estimate. We then compare the empirical covariance of the resulting estimates with the Cram\'er--Rao lower bound predicted by the SFIM. The smooth FIM has rank six for the seven-parameter problem and therefore does not yield a finite Cram\'er--Rao bound in all parameter directions, whereas the SFIM is full rank.

\begin{table}[b!]
\centering
\caption{Monte Carlo parameter-estimation variance
and SFIM Cram\'er--Rao lower bound}
\label{tab:mc_crlb}
\resizebox{\columnwidth}{!}{
\begin{tabular}{c|ccc}
\hline
Parameter & MC variance & SFIM CRLB & MC/CRLB \\
\hline
$M_r$     & $3.43\times10^{-5}$  & $3.77\times10^{-5}$  & $0.909$ \\
$\kappa$  & $3.61\times10^{-7}$  & $3.88\times10^{-7}$  & $0.929$ \\
$k_p$     & $5.23$                & $5.36$                & $0.975$ \\
$k_i$     & $2.84$                & $2.45$                & $1.162$ \\
$T_\beta$ & $1.81\times10^{-2}$  & $2.03\times10^{-2}$  & $0.888$ \\
$M$       & $6.73\times10^{-10}$ & $6.69\times10^{-10}$ & $1.007$ \\
$T_E$     & $1.68\times10^{-6}$  & $1.62\times10^{-6}$  & $1.033$ \\
\hline
\end{tabular}}
\end{table}

As shown in Table~\ref{tab:mc_crlb}, the empirical Monte Carlo variances closely follow the corresponding diagonal entries of SFIM. Across the seven parameters, the mean ratio of the empirical variance to the SFIM Cram\'er--Rao bound is $0.986$, with a median ratio of $0.975$. This close agreement supports the SFIM's predictive capability for local parameter-estimation uncertainty. In contrast, the rank deficiency of the smooth FIM indicates that the corresponding smooth approximation cannot provide a finite covariance bound in every parameter direction.

Fig. \ref{fig:eigvec} shows the eigenvector associated with $\lambda_{\min}(F)$ under SFIM. The weakest information direction lies predominantly in the $(k_p,k_i)$ subspace with opposite signs, revealing a compensating PI gain trade-off. To assess the information available for these gains, consider the conditional Fisher information. Partition the parameter vector as $\theta =[\theta_a,\;\theta_b]^\top$ and $\theta_a = [k_p,\;k_i]^\top$, where $\theta_b$ contains the remaining parameters. The effective information about $\theta_a$ after accounting for $\theta_b$ is given by the Schur complement $F_{a|b} = F_{aa} - F_{ab}F_{bb}^{-1}F_{ba}$. For SFIM, the eigenvalues of $F_{a|b}$ are $\lambda(F_{a|b}) = [0.156,\;0.712]$, which are strictly positive. Thus, both $k_p$ and $k_i$ are locally identifiable, although their strong negative correlation indicates a weakly excited compensating gain direction. 

Although the reset map in \eqref{eq:hybridModel} is the identity, the sensitivity changes discontinuously through the saltation matrix at each guard crossing. The resulting increments satisfy $25 \le \Delta I_j \le 129$. 
\begin{figure}[!t]
\vspace{-2mm}
\centering
\includegraphics[width=1\linewidth]{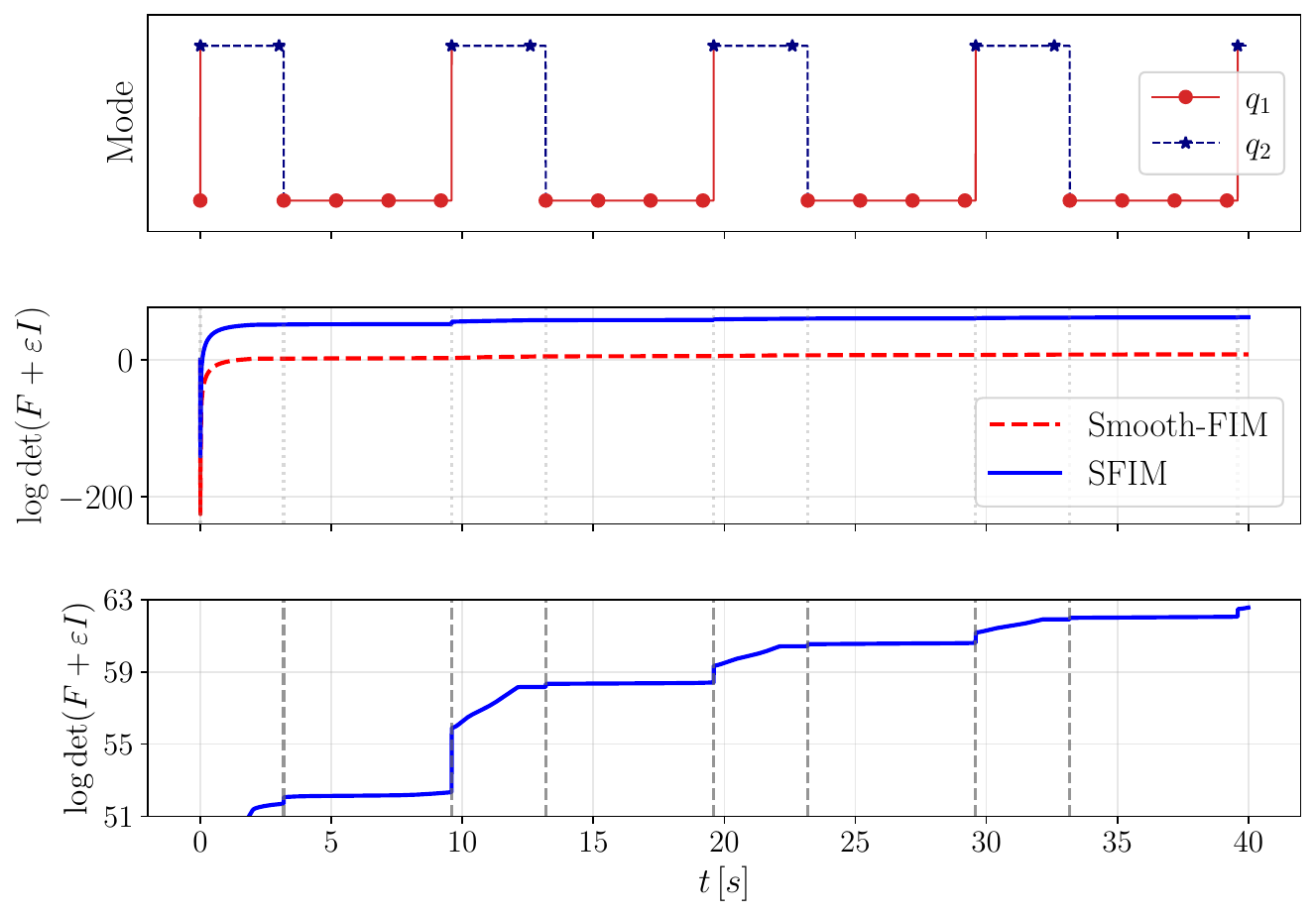}
\caption{Top: Hybrid mode sequence showing switching between $q_1$ and $q_2$ over time. Middle: Evolution of $\log \det(F+\varepsilon I)$ for the smooth FIM and the SFIM. Bottom: Zoomed-in view of the middle plot highlighting the stepwise increases in the SFIM at switching instants.}
\label{fig:fim_switch}
\vspace{-3mm}
\end{figure}

Next, we consider the power-triggered guard in \eqref{eq:powerGuard}. To examine the role of parameter-dependent event times, Proposition~\ref{prop:event_time_sensitivity} is applied to the two guard conditions.

Let $x=[E^{\prime}_{q},\omega,T_{m},E_{\mathrm{fd}},z,\xi,\beta]^\top$ and define the selector $e_z=[0\ 0\ 0\ 0\ 1\ 0\ 0]^\top$. For the rotor-speed guard, $D_x g_z = e_z$, $D_t g_z = 0$, $D_\theta g_z = 0$, and the event-time sensitivity reduces to
\begin{align*}
D_\theta \tau
=
-\frac{e_z^\top Z(\tau^-)}
{e_z^\top f_q^-(\tau^-)},
\end{align*}
which indicates the switching time depends on $\theta$ only through the state sensitivity.

For \eqref{eq:powerGuard}, $D_x g_{P_{\mathrm{ele}}} = D_x P_{\mathrm{ele}}$, $D_t g_{P_{\mathrm{ele}}} \neq 0$, and $D_\theta g_{P_{\mathrm{ele}}} = D_\theta P_{\mathrm{ele}} \neq 0$. In contrast to the guard condition in \eqref{eq:hybridModel}, the switching time depends explicitly on the parameters through $D_\theta P_{\mathrm{ele}}$ and on time through $D_t P_{\mathrm{ele}}$. Consequently, the event time $\tau(\theta)$ becomes parameter-dependent and contributes to the saltation-based sensitivity update, which is not captured by reset Jacobian approximations. This effect is reflected in Fig.~\ref{fig:eigvec} (right), where the smooth and reset-based formulations produce nearly identical spectra and contribute little additional information at switching, while the SFIM exhibits a different eigenvalue distribution.

\begin{figure}[!t]
\centering
\includegraphics[width=\linewidth]{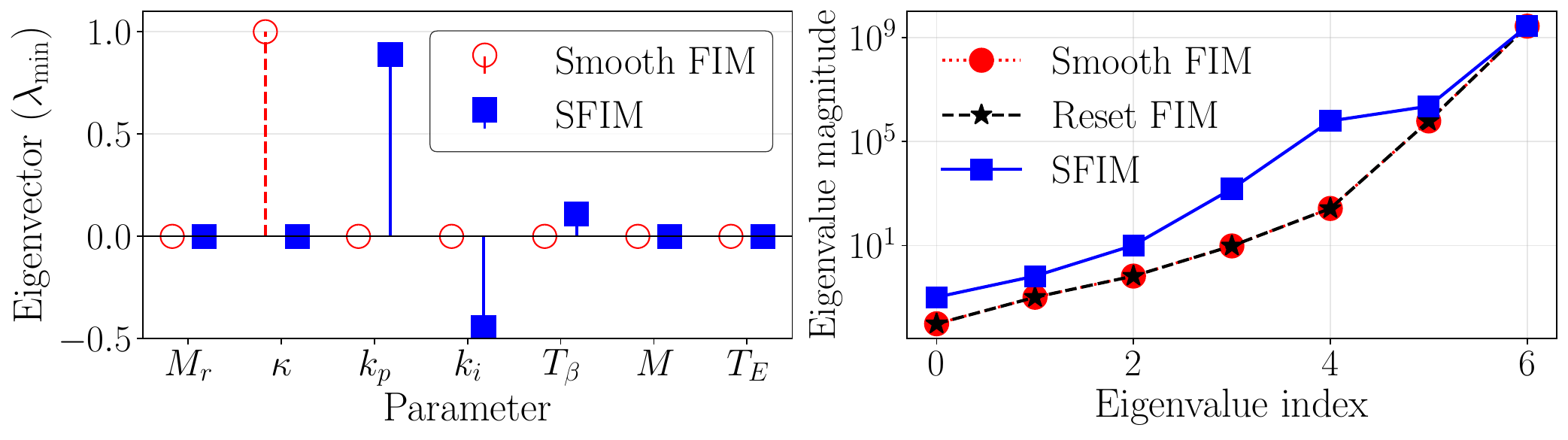}
\caption{Left: Least observable direction of the FIM under the $z$-guard for the smooth model and SFIM. Right: Eigenvalues of the FIM under the $P_{\mathrm{ele}}$ guard for the smooth model, reset Jacobian, and SFIM.}
\label{fig:eigvec}
\vspace{-4mm}
\end{figure}

\section{Conclusions}
This paper introduced a salted Fisher information matrix (SFIM) for hybrid systems. With saltation-based sensitivity updates at guard crossings, the formulation accounts for both continuous flows and discrete events in parameter-information propagation. We showed that hybrid persistence of excitation is a sufficient condition for the SFIM to be positive definite. Numerical experiments on a three-bus wind turbine generator system show that the SFIM yields a substantially richer information structure than smooth or reset-only approximations. Our formulation assumes deterministic hybrid dynamics with known guard and reset maps and relies on first-order sensitivity approximations. Extensions to stochastic hybrid systems and trajectory design for SFIM maximization remain directions for future work.

\bibliographystyle{IEEEtran} 
\bibliography{reference}

@article{odunlami2025hybrid,
  title={Hybrid Dynamical Systems Modeling of Power Systems},
  author={Odunlami, Bukunmi Gabriel and Netto, Marcos and Susuki, Yoshihiko},
  journal={preprint arXiv:2509.02822},
  year={2025}
}

@ARTICLE{KONG2021109752,
author={Kong, Nathan J. and Joe Payne, J. and Zhu, James and Johnson, Aaron M.},
journal={Proc. IEEE}, 
title={Saltation Matrices: The Essential Tool for Linearizing Hybrid Dynamical Systems}, 
year={2024},
volume={112},
number={6},
pages={585-608},}

@ARTICLE{HPE,
author={Saoud, Adnane and others},
journal={IEEE Trans. Autom. Control}, 
title={Hybrid Persistency of Excitation in Adaptive Estimation for Hybrid Systems}, 
year={2024},
volume={69},
number={12},
pages={8828-8835},}

@ARTICLE{observabilityofHybrid,
  author={Lin, Feng and others},
  journal={IEEE Trans. Autom. Control}, 
  title={On Observability of Hybrid Systems}, 
  year={2022},
  volume={67},
  number={11},
  pages={6074-6081},
  doi={10.1109/TAC.2021.3126610}}

@ARTICLE{JauffFisher,
  author={Jauffret, Claude},
  journal={IEEE Trans. Aerosp. Electron. Syst.}, 
  title={Observability and fisher information matrix in nonlinear regression}, 
  year={2007},
  volume={43},
  number={2},
  pages={756-759},
  doi={10.1109/TAES.2007.4285368}}

@ARTICLE{VirginFisherPowerSystems,
author={Virginillo, Dawn and Derviškadić, Asja and Paolone, Mario},
journal={IEEE Trans. Power Syst.}, 
title={Identification of Power System Dynamic Model Parameters Using the Fisher Information Matrix}, 
year={2026},
volume={41},
number={2},
pages={1376-1388},}

@article{Cavicchioli2017,
title = {Asymptotic Fisher information matrix of Markov switching VARMA models},
journal = {J. Multivar. Anal.},
volume = {157},
pages = {124-135},
year = {2017},
author = {Maddalena Cavicchioli},
}

@article{KON2021109752,
title = {The Salted Kalman Filter: Kalman filtering on hybrid dynamical systems},
journal = {Automatica},
volume = {131},
pages = {109752},
year = {2021},
author = {Nathan J. Kong and others},}

@article{Antsaklis,
author = {Lin, Hai and Antsaklis, Panos J.},
title = {Hybrid Dynamical Systems: An Introduction to Control and Verification},
year = {2014},
issue_date = {Mar 2014},
publisher = {Now Publishers Inc.},
address = {Hanover, MA, USA},
volume = {1},
number = {1},
issn = {2325-6818},
journal = {Found. Trends Syst. Control},
month = mar,
pages = {1–172},
numpages = {172}
}

@INPROCEEDINGS{Johnson,
  author={Johnson, Ryan S. and Di Cairano, Stefano and Sanfelice, Ricardo G.},
  booktitle={IEEE Conf. Decis. Control}, 
  title={Parameter Estimation for Hybrid Dynamical Systems using Hybrid Gradient Descent}, 
  year={2021},
  volume={},
  number={},
  pages={4648-4653},
}

@INPROCEEDINGS{5675187,
author={Vlad, Cristina and others},
booktitle={IECON 2010 - 36th Annu. Conf. IEEE Ind. Electron. Soc.}, 
title={A hybrid model for Buck converter operating in continuous and discontinuous conduction modes}, 
year={2010},
volume={},
number={},
pages={138-143},
}

@ARTICLE{Yu_reach,
author={Chen, Yu Christine and Dominguez-Garcia, Alejandro D.},
journal={IEEE Trans. Power Syst.}, 
title={A Method to Study the Effect of Renewable Resource Variability on Power System Dynamics}, 
year={2012},
volume={27},
number={4},
pages={1978-1989},
}

@article{WANG20111288,
title = {Synthesis on PI-based pitch controller of large wind turbines generator},
journal = {Energy Convers. Manag.},
volume = {52},
number = {2},
pages = {1288-1294},
year = {2011},
author = {Junsong Wang and Norman Tse and Zhiwei Gao},
}

@book{sauer_pai,
  author    = {Peter W. Sauer and M. A. Pai},
  title     = {Power System Dynamics and Stability},
  publisher = {Prentice-Hall},
  year      = {1998}
}

@article{Hartley,
author = {Hartley, Ross and Ghaffari, Maani and Eustice, Ryan and Grizzle, Jessy},
year = {2020},
pages = {027836491989438},
title = {Contact-aided invariant extended Kalman filtering for robot state estimation},
volume = {39},
journal = {Int. J. Robot. Res.},
}

@article{wilson2014fisher,
  author    = {Wilson, Andrew D. and Schultz, Jarvis A. and Murphey, Todd D.},
  title     = {Trajectory Synthesis for Fisher Information Maximization},
  journal = {IEEE Trans. Robot.},
  year      = {2014},
  volume    = {30},
  number    = {6},
  pages     = {1358--1370},
}

@INPROCEEDINGS{7039861,
  author={Saccon, Alessandro and van de Wouw, Nathan and Nijmeijer, Henk},
  booktitle={IEEE Conf. Decis. Control}, 
  title={Sensitivity analysis of hybrid systems with state jumps with application to trajectory tracking}, 
  year={2014},
  volume={},
  number={},
  pages={3065-3070},}

@article{IVANOV1998677,
title = {The stability of periodic solutions of discontinuous systems that intersect several surfaces of discontinuity},
journal={J. Appl. Math. Mech.},
volume = {62},
number = {5},
pages = {677-685},
year = {1998},
author = {A.P. Ivanov},
}

@book{Kay1993,
author = {Steven M. Kay},
title = {Fundamentals of Statistical Signal Processing: Estimation Theory},
publisher = {Prentice-Hall},
year = {1993},
volume = {1}
}

@INPROCEEDINGS{6638944,
  author={Pakrooh, Pooria and Scharf, Louis L. and Pezeshki, Ali and Chi, Yuejie},
  booktitle={IEEE Int. Conf. Acoust., Speech, Signal Process.},
  title={Analysis of fisher information and the Cramer-Rao bound for nonlinear parameter estimation after compressed sensing}, 
  year={2013},
  volume={},
  number={},
  pages={6630-6634},
 }

@ARTICLE{9585016,
  author={Fabbiani, Emanuele and Nahata, Pulkit and De Nicolao, Giuseppe and Ferrari-Trecate, Giancarlo},
  journal={IEEE Trans. Control Syst. Technol.}, 
  title={Identification of AC Distribution Networks With Recursive Least Squares and Optimal Design of Experiment}, 
  year={2022},
  volume={30},
  number={4},
  pages={1750-1757},
 }

@article{Sovljanski2024,
  author  = {V. Sovljanski and M. Paolone},
  title   = {On the use of Cramer--Rao lower bound for least-variance circuit parameters identification of Li-ion cells},
  journal = {J. Energy Storage},
  volume  = {94},
  pages   = {112223},
  year    = {2024}
}

@INPROCEEDINGS{Koutsoukos2003,
author={Koutsoukos, X.D.},
booktitle={IEEE Conf. Decis. Control}, 
title={Estimation of hybrid systems using discrete sensors}, 
year={2003},
volume={1},
number={},
pages={155-160},}

@ARTICLE{Tichavsky1998,
author={Tichavsky, P. and Muravchik, C.H. and Nehorai, A.},
journal={IEEE Trans. Signal Process.}, 
title={Posterior Cramer-Rao bounds for discrete-time nonlinear filtering}, 
year={1998},
volume={46},
number={5},
pages={1386-1396},}
\end{document}